\documentclass[a4paper,12pt]{article}
\usepackage{latexsym,amssymb,amsfonts,amsmath,amsthm}
\usepackage[dvips]{graphicx}
\usepackage{comment}

\newtheorem{theorem}{Theorem}
\newtheorem{lemma}{Lemma}
\newtheorem{corollary}{Corollary}
\newtheorem{proposition}{Proposition}
\newtheorem{remark}{Remark}
\newtheorem{definition}{Definition}

\numberwithin{theorem}{section}
\numberwithin{lemma}{section}
\numberwithin{corollary}{section}
\numberwithin{proposition}{section}
\numberwithin{remark}{section}

\newcommand{\spann}{{\rm span}}

\newcommand{\supp}{{\rm supp}}

\newcommand{\Deg}{{\rm deg}}
\newcommand{\dist}{{\rm dist}}

\usepackage{color}

\title{{\Large {\bf Relation between quantum walks with tails and quantum walks with sinks on finite graphs}}
\author{
{\small Norio Konno}\\
{\scriptsize Department of Applied Mathematics, Yokohama National University, }\\  
{\scriptsize Hodogaya, Yokohama, 240-8501, Japan. }\\
{\small Etsuo Segawa\footnote{Email: segawa-etsuo-tb@ynu.ac.jp}}\\
{\scriptsize Graduate School of Environment Information Sciences, Yokohama National University, } \\
{\scriptsize Hodogaya, Yokohama, 240-8501, Japan. }\\
{\small Martin \v{S}tefa\v{n}\'{a}k}\\
{\scriptsize Department of Physics, Faculty of Nuclear Sciences and Physical Engineering, Czech Technical University in Prague, }\\
{\scriptsize B\v{r}ehov\'{a} 7, 115 19 Praha 1 - Star\'{e} M\v{e}sto, Czech Republic}
}
}
\date{\empty }
\begin{document}
\maketitle

\par\noindent
\begin{small}
\par\noindent
{\bf Abstract}. We connect the Grover walk with sinks to the Grover walk with tails. The survival probability of the Grover walk with sinks in the long time limit is characterized by the centered generalized eigenspace of the Grover walk with tails. The centered eigenspace of the Grover walk is the attractor eigenspace of the Grover walk with sinks. It is described by the persistent eigenspace of the underlying random walk whose support has no overlap to the boundaries of the graph and combinatorial flow in the graph theory.

\footnote[0]{
{\it Keywords: } 
Quantum walk, survival probability, attractor eigenspace, dressed photon 
}
\end{small}

\section{Introduction}
A simple random walker on a finite and connected graph starting from any vertex hits an arbitrary vertex in a finite time. This fact implies that if we consider a subset of the vertices of this graph as sinks, where the random walker is absorbed, then the survival probability of the random walk in the long time limit converges to zero. However, for quantum walks (QW) \cite{Ambainis2003} the situation is more complicated and the survival probability depends in general on the graph, coin operator and the initial state of the walk. For a two-state quantum walk on a finite line with sinks on both ends and a non-trivial coin the survival probability is also zero, as shown by the studies of the corresponding absorption problem~\cite{Ambainis,Konno,Bach,Yamasaki}.
However, for a three-state quantum walk with the Grover coin \cite{IKS}  the survival probability on a finite line is non-vanishing \cite{SNJ} due to the existence of trapped states. These are the eigenstates of the unitary evolution operator which do not have a support on the sinks. Trapped states crucially affect the efficiency of quantum transport \cite{MNJ} and lead to counter-intuitive effects, e.g. the transport efficiency can be improved by increasing the distance between the initial vertex and the sink \cite{MNSJ,MNJ:2020}. We find a similar phenomena to this quantum walk model in the experiment 
on the energy transfer of the dressed photon~\cite{DressedPhoton0} through 
the nanoparticles distributed in a finite three dimensional grid~\cite{DressedPhoton1}. The output signal intensity increases when the depth direction is larger.   
Although when the depth is deeper, 
a lot of ``detours" newly appear to reach to the position of the output from the classical point of view, 
the output signal intensity of the dressed photon becomes stronger.
The existence of trapped states also results in infinite hitting times \cite{Krovi:hypercube,Krovi:infhit}. 

In this paper we analyse such counter-intuitive phenomena for the Grover walk on general connected graph using the spectral analysis. 
The Grover walk is an induced quantum walk of the random walk from the view point of the spectral mapping theorem~\cite{HKSS}.   

To this end, first we connect the Grover walk with sink to the Grover walk with tails.  The tails are the semi-infinite paths attached to a finite and connected graph. We call the set of vertices connecting to the tails the boundary.  The Grover walk with tail is introduced by \cite{FH1,FH2} in terms of  the scattering theory. If we set some appropriate bounded initial state so that the support is included in the tail, the existence of the fixed point of the dynamical system induced by the Grover walk with tails is shown, and the stable generalized eigenspace $\mathcal{H}_s$, in which the dynamical system lives, is orthogonal to the centered generalized eigenspace $\mathcal{H}_c$~\cite{R} at every time  step~\cite{HS}. 
The centered generalized eigenspace is generated by the generalized eigenvectors of the principal submatrix of the time evolution operator of the Grover walk with respect to the internal graph, and all the corresponding absolute values of the eigenvalues are $1$. 
This eigenstate is equivalent to the attractor space~\cite{MNJ} of the Grover walk with sink. 
Indeed, we show that the stationary state of the Grover walk with sink is attracted to this centered generalized eigenstate. 
Secondly, we characterize this centered generalized eigenspace using the persistent eigenspace of the underlying random walk whose supports have no overlaps to the boundary and also using the concept of ``flow" from the graph theory. From this result, we see that the existence of the persistent eigenspace of the underlying random walk influences significantly the asymptotic behavior of the corresponding Grover walk, although it has little effect on the asymptotic behavior of the random walk itself. Moreover, we clarify that the graph structure which constructs the symmetric or anti-symmetric flow satisfying the Kirchhoff's law contributes to the non-zero survival probability of the Grover walk as suggested by \cite{HKSS,MNJ}.    

This paper is organized as follows. 
In section 2, we prepare the notations of graphs and give the definition of the Grover walk and the boundary operators which are related to the chain.  
In section 3, we give the definition of the Grover walk on a graph with sinks. In section 4, a necessary and sufficient condition for the surviving of the Grover walk are described. 
In section 5, we give an example. 
Section 6 is devoted to the relation between the Grover walk with sink and the Grover walk with tail. 
In section 7, we partially characterize  the centered generalized eigenspace using the concept of flow from the graph theory.  
\section{Preliminary}
\subsection{Graph notation}
Let $G=(V,A)$ be a connected and {\it symmetric digraph} such that an arc $a\in A$ if and only if its inverse arc $\overline{a}\in A$. The {\it origin and terminal vertices} of $a\in A$ are denoted by $o(a)\in V$ and $t(a)\in V$, respectively. Assume that $G$ has no multiple arcs. 
If $t(a)=o(a)$, we call such an arc $a$ the {\it self-loop}. In this paper, we regard  $\overline{a}=a$ for any self-loops.  We denote $A_{\sigma}$ as the set of all the self-loops.  
The {\it degree} of $v\in V$ is defined by 
$$
\Deg(v)=|\{a\in A \;|\; t(a)=v\}|.
$$
The {\it support edge} of $a\in A\setminus A_{\sigma}$ is denoted by $|a|$ with $|a|=|\overline{a}|$. 
The set of {\it (non-directed) edges} is 
$$
E=\{|a| \;|\; a\in A\setminus A_{\sigma}\}.
$$ 
A {\it walk} in $G$ is a sequence of arcs such that $p=(a_0,a_1,\dots,a_{r-1})$ with $t(a_j)=o(a_{j+1})$ for any $j=0,\dots,r-2$, which may have the same arcs in $p$. 
The {\it cycle} in $G$ is a subgraph of $G$ which is isomorphic to a sequence of arcs $(a_0,a_1,\dots,a_{r-1})$ ($r\geq 3$) satisfying $t(a_j)=o(a_{j+1})$ with $a_j\neq \overline{a}_{j+1}$ for any $j=0,\dots,r-1$, where the subscript is the modulus of $r$. We identify $(a_{k},a_{k+1},\dots,a_{k+r-1})$ with $(a_0,a_1,\dots,a_{r-1})$ for  $k\in\mathbb{Z}$.
The {\it spannig tree} of $G$ is a connected subtree of $G$ covering all vertices of $G$.
A {\it fundamental cycle} induced by the spanning tree is the cycle in $G$ generated by recovering an arc which is outside of the spanning tree to the spanning tree. 
There are two choices of orientations for each support of the fundamental cycle, but we choose only one of them as the representative.  
Fixing a spanning tree, we denote the set of fundamental cycles by $\Gamma$. Then the cardinality of $\Gamma$ is $|E|-|V|+1=:b_1$.  We call $b_1$ the {\it first Betti number}. 

\subsection{Definition of the Grover walk}
Let $\Omega$ be a discrete set. The vector space whose standard basis is labeled by each element of $\Omega$ is denoted by $\mathbb{C}^\Omega$. The standard basis is denoted by $\delta_\omega^{(\Omega)}$ ($\omega\in \Omega$), i.e., 
    \[ \delta_\omega^{(\Omega)}(\omega')=\begin{cases} 1 & \text{: $\omega=\omega'$,} \\ 0 & \text{: otherwise.} \end{cases} \]
Throughout this paper, the inner product is standard, i.e., 
$$
\langle \psi,\phi\rangle_{\Omega}=\sum_{\omega\in \Omega}\bar{\psi}(\omega)\phi(\omega),
$$
for any $\psi,\phi\in \mathbb{C}^\Omega$, and the norm is defined by 
$$
||\psi||_\Omega=\sqrt{\langle \psi,\psi\rangle_{\Omega}}.
$$
For any $\psi\in \mathbb{C}^{\Omega}$, the support of $\psi$ is defined by 
$$
\supp(\psi):=\{ \omega\in \Omega \;|\; \psi(\omega)\neq 0\}.
$$
For subspaces $M,N\subset \mathbb{C}^{\Omega}$, the relation 
$$
\mathbb{C}^{\Omega}=M\oplus N,
$$
means that $M$ and $N$ are complementary spaces in $\mathbb{C}^{\Omega}$, i.e., for any $f\in \mathbb{C}^{\Omega}$, $g\in M$  and $h\in N$ are uniquely determined such that $f=g+h$; which means if $u'+u''=0$ for some $u'\in \Omega'$ and $u''\in \Omega''$, then $u'$ and $u''$ must be $u'=u''=0$. Note that $\langle g,h \rangle_\Omega\neq 0$ in general, i.e., $M$ and $N$ are not necessarily orthogonal subspaces. Especially in this paper, we treat an operator which is a submatrix of a unitary operator, and we are not ensured that it is a normal operator. 
The vector space describing the whole system of the Grover walk is $\mathbb{C}^A$.
The time evolution operator of the Grover walk on $G$ is defined by 
\[ (U_G\psi)(a)=-\psi(\overline{a})+\frac{2}{\deg(o(a))}\sum_{t(b)=o(a)}\psi(b) \]
for any $\psi\in \mathbb{C}^A$ and $a\in A$. Note that since $U_G$ is a unitary operator on $\mathbb{C}^A$, $U_G$ preserves the $\ell^2$ norm, i.e., $||U_G\psi||_A^2=||\psi||_A^2$. 
Let $\psi_n\in \mathbb{C}^A$ be the $n$-th iteration of the Grover walk $\psi_n=U_G\psi_{n-1}$ ($n\geq 1$) with the initial state $\psi_0$.  
Then the probability distribution at time $n$, $\mu_n: V\to [0,1]$, can be defined by 
    \[ \mu_n(v)=\sum_{t(a)=v}|\psi_n(a)|^2 \]
if the norm of the initial state is unity.  
Our interest is the asymptotic behavior of the sequence of probabilities $\mu_n$ and also of amplitudes $\psi_n$ on the graph comparing with the behavior of the corresponding random walk. 
\subsection{Boundary operators}
Let $G=(V,A)$ be the original graph. The set of sinks is denoted by $V_s\subset V$. The subgraph of $G$; $G_0=(V_0,A_0)$, is defined by  
\[ V_0=V\setminus V_s,\;A_0=\{a\in A \;|\; t(a),o(a)\notin V_s\}. \]
The set of self-loops in $G_0$ is denoted by $A_{0,\sigma}\subset A_0$. 
See Fig~\ref{Fig:2}. 
The set of the fundamental cycles in $G_0$ is denoted by $\Gamma$ hereafter. 
The set of boundary vertices of $G_0$ is defined by \[ \delta G_0=\{ o(a) \;|\; a\in A,\; o(a)\in V\setminus V_s,\; t(a)\in V_s \}.  \]
Under the above settings of graphs, let us now prepare some notations to show our main theorem. 
\begin{definition}
Let $\tilde{d}(u)$ be the degree of $u$ in the original graph $G$.
Let $G_0=(V_0,A_0)$ be the subgraph as above. 
Then the boundary operators $d_1:\mathbb{C}^{A_0}\to \mathbb{C}^{V_0}$ and $\partial_2:\mathbb{C}^{\Gamma}\to \mathbb{C}^{A_0}$ are denoted by  
    \[ (d_1\psi)(v)=\frac{1}{\sqrt{\tilde{d}(v)}}\sum_{t(a)=v}\psi(a),\;\;(\partial_2\Psi)(a)=\sum_{a\in A(c)\subset A_0}\Psi(c), \]
respectively, for any $\psi\in\mathbb{C}^A$, $\Psi\in\mathbb{C}^{\Gamma}$ and $v\in V_0$, $a\in A_0$. Here $A(c)$ is the set of arcs of $c\in \Gamma$.  
\end{definition}
Note that $\tilde{d}(u)$ is the degree of $G$, so if $u\in \delta G_0$, then $\tilde{d}(u)$ is greater than the degree in $G_0$.  
The adjoint operators of $d_1$ and $\partial_2$ are defined by
    \[ \langle f, d_1\psi \rangle_{V_0} = \langle d_1^*f, \psi \rangle_{A_0},\quad
    \langle \psi, \partial_2\Psi \rangle_{A_0} = \langle \partial_2^*\psi, \Psi \rangle_\Gamma \]
which imply
    \[(d_1^*f)(a)=f(t(a)),\; (\partial_2^*\psi)(c)=\sum_{a\in A(c)}\psi(a). \]
Let $S: \mathbb{C}^{A_0}\to \mathbb{C}^{A_0}$ be a unitary operator defined by $(S\psi)(a)=\psi(\overline{a})$.
We prove that the composition of $d_1(I-S)\circ \partial_2$ is identically equal to zero as follows.  
\begin{lemma}\label{lem:boundaries}
Let $d_1$ and $\partial_2$ be the above. Then we have 
    \[ d_1(I-S)\partial_2=0. \]
\end{lemma}
\begin{proof}
For any $c\in \Gamma$, let $\delta_c^{(\Gamma)}\in \mathbb{C}^{\Gamma}$ be the delta function, i.e., 
    \[ \delta_c^{(\Gamma)}(c')=\begin{cases} 1 & \text{: $c=c'$,}\\ 0 & \text{: $c\neq c'$.}\end{cases} \] 
Then it is enough to see that $d_1(I-S)\partial_2\delta_c^{(\Gamma)}=0$ for any $c\in \Gamma$. Indeed, we find 
    \begin{align*}
        d_1(I-S)\partial_2\delta_c^{(\Gamma)} 
        &= d_1(\sum_{a\in A(c)}\delta_a^{(A)}-\sum_{a\in A(c)}\delta_{\overline{a}}^{(A)}) \\
        &= \sum_{a\in A(c)}\frac{1}{\sqrt{\tilde{d}(t(a))}} \delta_{t(a)}^{(V)}-\sum_{a\in A(c)}\frac{1}{\sqrt{\tilde{d}(t(\overline{a}))}}\delta_{t(\overline{a})}^{(V)} \\
        &= 0,
    \end{align*}
which is the desired conclusion. 
\end{proof}
Let us set the function $\xi_c^{(+)}$ induced by $c\in \Gamma$ by 
     \[ \xi_c^{(+)}:= (I-S)\partial_2\delta_c^{(\Gamma)}.  \]
In other words, $\supp(\xi_c^{(+)})=A(c)\cup A(\bar{c})$ and 
    \[ (\xi_c^{(+)})(a)=\begin{cases} 1 & \text{: $a\in A(c)$,}\\ 
    -1 & \text{: $\overline{a}\in A(c)$,}\\ 0 & \text{: otherwise.}\end{cases} \]
Let us introduce $\chi_S: \mathbb{C}^{A}\to \mathbb{C}^{A_0}$ by  
    \[ (\chi_S\phi)(a)=\phi(a) \]
for all $a\in A_0$. The adjoint $\chi_S^*: \mathbb{C}^{A_0}\to \mathbb{C}^A$ is described by 
    \[ (\chi_S^*f)(a)=\begin{cases} f(a) & \text{: $a\in A_0$,}\\ 0 & \text{: otherwise.} \end{cases} \]
The function $\xi_c^{(+)}$ satisfies the following properties:
\begin{proposition}\noindent
For any fundamental cycle $c$ in $G_0\subset G$, we have 
$\chi_S^* \xi_c^{(+)}\in \ker(1-U_G)$.
\end{proposition}
\begin{proof}\noindent
The following direct computation gives the consequence: 
    \begin{align*}
        (U_G\chi_S^*\xi_c^{(+)})(a) &= -(\chi_S^*\xi_c^{(+)})(\overline{a})+\frac{2}{\tilde{d}(o(a))}\sum_{t(b)=o(a)} (\chi_S^*\xi_c^{(+)})(b) \\
        &= (\chi_S^*\xi_c^{(+)})(a) + \frac{2}{\sqrt{\tilde{d}(o(a))}} (d_1\chi_S^*\xi_c^{(+)})(o(a)) \\
        &= (\chi_S^*\xi_c^{(+)})(a).
    \end{align*}
Here the first equality derives from the definition of $U_G$. In the second equality, since $\supp( \xi_c^{(+)})\subset A_0\subset A$ and the summation of RHS in the first equality are essentially the same as the one over $A_0$, we can apply the definition of $d_1$ to this. 
We used Lemma~\ref{lem:boundaries} in the last equality.   
\end{proof}
We set $\mathcal{K}\subset \mathbb{C}^{A_0}$ by 
\begin{equation}\label{eq:defK}
    \mathcal{K}=\mathrm{span}\{\chi_S \xi_c^{(+)} \;|\; c\in \Gamma\subset G_0\}.
\end{equation}  
The self-adjoint operator 
$$
T:=(\chi_S d_1)S(\chi_S d_1)^* 
$$
on $\mathbb{C}^{A_0}$ is isomorphic to the transition probability operator $P'$ with the Dirichlet boundary condition on $\delta V_0$; i.e, 
$$
P'=D^{-1/2}TD^{1/2},
$$
where $(Df)=\tilde{d}(u)f(u)$. 
Here the matrix representation of $P'$ is described by  \[(P')_{u,v}:= \langle \delta_u^{(V_0)},P'\delta_v^{(V_0)}\rangle_{V_0}=\begin{cases} 1/\tilde{d}(u) & \text{: if $u$ and $v$ are connected,}\\ 0 & \text{: otherwise,}\end{cases}\] 
for any $u,v\in V_0$. 
If $Tf=x f$ and $Tg=y g$ ($x\neq y$), then we find the orthogonality such that 
    \[ \langle (1-e^{i \arccos x} S)d_1^*f,(1-e^{-i \arccos y} S)d_1^*g\rangle =0, \]
    \[\langle (1-e^{i \arccos x} S)d_1^*f,(1-e^{i \arccos y} S)d_1^*{g}\rangle =0,\]
    \[\langle (1-e^{i \arccos x} S)d_1^*f,(1-e^{-i \arccos y} S)d_1^*{g}\rangle =0. \]
Then we set $\mathcal{T}\subset \mathbb{C}^{A_0}$ by 
    \begin{equation}\label{eq:RW} 
    \mathcal{T}= \bigoplus_{|\lambda|=1} \{ (1-\lambda S)d_1^*f \;|\; f\in   \ker((\lambda+\lambda^{-1})/2-T),\;\supp(f)\subset V_0\setminus \delta V_0\}.
    \end{equation} 
This is the subspace of $\mathbb{C}^{A_0}$ lifted up from the eigenfunctions in $\mathbb{C}^{V_0}$ of the Dirichlet cut random walk $T$ by $(1-\lambda S)d_1^*f$. 
It is shown that $\mathrm{Spec}(E)\subset \mathbb{D}$ where $\mathbb{D}$ is the unit disc $\{z\in \mathbb{C} \;|\; |z|\leq 1\}$ in Proposition~\ref{prop:stationary}, and $\mathcal{T}=\oplus_{|\lambda|=1,\;\lambda\neq \pm 1}\ker(\lambda-E)$, where $E:=\chi_S U_G \chi_S^*$ in Lemma~\ref{lem:HS}. 
\section{Definition of the Grover walk on graphs with sinks}
Let $G=(V,A)$ be a finite and connected graph with sinks $V_s=\{v_1,\dots,v_q\}\subset V$. 
We set the graph $G_0=(V_0,A_0)$ by $A(G_0)=\{a\in A \;|\; t(a),o(a)\notin V_s\}$ and $V(G_0)=V\setminus V_s$. 
Assume that $G_0:=G\setminus V_s$ is connected. For simplicity, in this paper we consider the initial state of the Grover walk $\phi_0$ that satisfies the condition $\supp(\phi_0)\subset A_0$.\footnote{If we consider general initial state $\phi_0'$ such that $\supp(\phi_0')\cap (A\setminus A_0)\neq \emptyset$, replacing  $\phi'_0$ into $\phi_0=\phi_1'$, we can reproduce the QW with this initial state after $n\geq 1$ by our setting.  }  
The time evolution of the Grover walk with sinks $V_s$ with such an initial state $\phi_0$ is defined by
\begin{equation}\label{eq:TEsink} 
\phi_n(a)= \begin{cases}
(U_G\phi_{n-1})(a) & \text{: $a\in V\setminus V_s$,} \\
0 & \text{: $a\in \partial V$,} \end{cases} 
\end{equation}
This means that a quantum walker at a sink falls into a pit trap. 
We are interested in the survival probability of the Grover walk defined by 
    \[ \gamma:= \lim_{n\to\infty}\sum_{a\in A}|\phi_n(a)|^2. \]
It is the probability that the quantum walker remains in the graph without falling into the sinks forever. 
Considering the corresponding isotropic random walk with sinks such that
    \[ p_n(v)=\begin{cases} (Pp_{n-1})(v) & \text{: $v\in V\setminus V_s$,}\\ 0 & \text{: $v\in V_s$,} \end{cases} \]
we find that its survival probability is zero
    \[ \gamma^{RW}:=\lim_{n\to\infty} \sum_{v\in V}p_n(v) = 0, \]
because the first hitting time of a random walk to an arbitrary vertex for a finite graph is finite. 
On the other hand, in the case of the Grover walk the survival probability becomes positive, up to the initial state. 
In this paper, we clarify a necessary and sufficient condition for $\gamma>0$. 
\section{Main theorem}
We consider the case study on $G_0$ by 
\begin{description}
    \item[Case A:] $A_{0,\sigma}=\emptyset$ and $G_0$ is a bipartite graph; 
    \item[Case B:] $A_{0,\sigma}=\emptyset$ and $G_0$ is a non-bipartite graph; 
    \item[Case C:] $A_{0,\sigma}\neq \emptyset$ and $G_0\setminus A_{0,\sigma}$ is a bipartite graph;
    \item[Case D:] $A_{0,\sigma}\neq \emptyset$ and $G_0\setminus A_{0,\sigma}$ is a non-bipartite graph.
\end{description}
For a subspace $\mathcal{H}\subset \mathbb{C}^{A_0}$, the projection operator onto $\mathcal{H}$ is denoted by $\Pi_{\mathcal{H}}$. 
Then we obtain the following theorem.
\begin{theorem}\label{thm:main}
Let $\phi_n$ be the $n$-th iteration of the Grover walk on $G=(V,A)$ with sinks.
Let the survival probability at time $n$ be defined by 
    \[ \gamma_n=\sum_{a\in A}|(\phi_n)|^2. \]
The subspaces $\mathcal{A},\mathcal{B},\mathcal{C},\mathcal{D}$ of $\mathbb{C}^{A_0}$ are defined in (\ref{A}),...,(\ref{D}), respectively.  
Then we have 
\begin{enumerate}
    \item $\lim_{n\to\infty}\gamma_n=\gamma$ exists;  
    \item The survival probability $\gamma$ is expressed by
    \[ \gamma = || \Pi_{\mathcal{T}}\chi_S\phi_0 ||^2 + || \Pi_{\mathcal{K}}\chi_S\phi_0 ||^2+
    \begin{cases}
    || \Pi_{\mathcal{A}}\chi_S\phi_0 ||^2 & \text{: Case A} \\
    || \Pi_{\mathcal{B}}\chi_S\phi_0 ||^2 & \text{: Case B} \\
    || \Pi_{\mathcal{C}}\chi_S\phi_0 ||^2 & \text{: Case C} \\
    || \Pi_{\mathcal{D}}\chi_S\phi_0 ||^2 & \text{: Case D} 
    \end{cases}
    \]
\end{enumerate}
\end{theorem}
\begin{proof}
Part 1 of Theorem~\ref{thm:main} is obtained by consequences of Proposition~\ref{prop:stationary} and Part 2 derives from Propositions~\ref{prop:h+}, \ref{prop:h-}.  
\end{proof}
From this Theorem, we obtain useful sufficient conditions for non-zero survival probability as follows. 
\begin{corollary}
Assume $G_0$ is a finite and connected graph. 
If $G_0$ is not a tree or $G_0$ has more than $2$ self-loops, then $\gamma>0$. 
\end{corollary}
\begin{remark}
The eigenspaces $\mathcal{A},\mathcal{B},\mathcal{C},\mathcal{D}$ correspond to the p-attractors defined in \cite{MNJ}.
\end{remark}

\section{Example}
\label{sec:ex}
Let us consider a simple example in Fig.~\ref{Fig:2}. 
$G_0=(V_0,A_0)$ with $V_0=\{u_1,u_2,u_3,u_4\}$ and $A_0=\{a_1,a_2,a_3,a_4,\overline{a}_1,\overline{a}_2,\overline{a}_3,\overline{a}_4,b_1,b_2\}$ with $u_1=t(a_4)=o(a_1)$, $u_2=t(a_1)=o(a_2)$, $u_3=t(a_2)=o(a_3)$, $u_4=t(a_3)=o(a_4)$ and $o(b_1)=t(b_1)=u_1$, $o(b_3)=t(b_3)=u_3$. 
This graph fits into Case C. 
So let $q$ be the closed walk by $q=(a_1,a_2,a_3,a_4)$ and 
$q'$ be the walk between two selfloops by $(b_1,a_1,a_2,b_2)$. 
Then 
\begin{align*}
    \xi_q^{(+)} &= (\delta_{a_1}+\delta_{a_2}+\delta_{a_3}+\delta_{a_4})
    -(\delta_{\overline{a}_1}+\delta_{\overline{a}_2}+\delta_{\overline{a}_3}+\delta_{\overline{a}_4},) \\
    \xi_q^{(-)} &=(\delta_{a_1}+\delta_{\overline{a}_1})-(\delta_{a_2}+\delta_{\overline{a}_2})+(\delta_{a_3}+\delta_{\overline{a}_3})-(\delta_{a_4}
    +\delta_{\overline{a}_4}), \\
    \eta_{b_1-b_2} &= \delta_{b_1}-(\delta_{a_1}+\delta_{\overline{a}_1})+(\delta_{a_2}+\delta_{\overline{a}_2})-\delta_{b_2}.
\end{align*}
The matrix representation of the self adjoint operator $T$ is expressed by 
    \[ T=\frac{1}{3}\begin{bmatrix}
    1 & 1 & 0 & 1 \\
    1 & 0 & 1 & 0 \\
    0 & 1 & 1 & 1 \\
    1 & 0 & 1 & 0
    \end{bmatrix}. \]
The eigenvector of $T$ which has no overlaps to $\delta V_0= \{2,4\}$ is  easily obtained by 
    \[ f=[1,\;0,\;-1,\;0]^\top \]
which satisfies $Tf=(1/3)f$. 
Here the symbol ``$\top$" is the transpose. 
The eigenfunctions lifted up to $\mathbb{C}^A$ from $f$ is 
    \[ (\varphi_\pm )(a)= f(t(a)) -\lambda_\pm f(o(a)) \]
by (\ref{eq:RW}), where 
$$
\lambda_\pm= \frac{1}{3}(1  \pm i \sqrt{8}) = e^{\pm i \theta}, \quad \theta = \arccos\frac{1}{3} .
$$
Then we have  
    \begin{align*} 
    \varphi_\pm(a_1)=-\lambda_\pm,\;\varphi_\pm(a_2)=-1,\;\varphi_\pm(a_3)=\lambda_\pm,\;\varphi_\pm(a_4)=1,\\
    \varphi_\pm(\bar{a}_1)=1,\;\varphi_\pm(\bar{a}_2)=\lambda_\pm,\;\varphi_\pm(\bar{a}_3)=-1,\;\varphi_\pm(\bar{a}_4)=-\lambda_\pm,\\
    \varphi_\pm(b_1)=1-\lambda_\pm,\;\varphi_\pm(b_2)=-1+\lambda_\pm.
    \end{align*}
It holds that $E\varphi_\pm=\lambda_\pm  \varphi_\pm$.  
We obtain 
    \begin{align*}
    \mathcal{T} &= \mathbb{C}\varphi_{+}\oplus \mathbb{C}\varphi_{-}, \\
    \mathcal{K} &= \mathbb{C}\xi_{(a_1,a_2,a_3,a_4)}^{(+)} , \\
    \mathcal{C} &= 
    \mathbb{C}\xi_{(a_1,a_2,a_3,a_4)}^{(-)}\oplus\mathbb{C}\eta_{b_1-b_2} . 
    \end{align*}
After the Gram Schmidt procedure to $\mathcal{C}$, we have 
    \[\mathcal{C}=\mathbb{C}\xi_{(a_1,a_2,a_3,a_4)}^{(-)}\oplus \mathbb{C}(\eta_{b_1-b_2}+\eta_{b_1-b_2}') . \]
Here we have denoted 
$$
\eta_{b_1-b_2}'=\delta_{b_1}-(\delta_{a_4}+\delta_{\overline{a}_4})+(\delta_{a_3}+\delta_{\overline{a}_3})-\delta_{b_2},
$$
see Fig.~\ref{Fig:3}; we express the functions $\varphi_\pm$, $\xi^{(+)}_{(a_1,a_2,a_3,a_4)}$, $\eta_{b_1-b_2}$, $\eta_{b_1-b_2}'$, $\eta_{b_1-b_2}+\eta_{b_1-b_2}'$ by weighted sub-digraphs of $G_0$.  
Then the time evolution of the asymptotic dynamics of this quantum walk is described by 
\begin{multline} 
U^n\sim 
\frac{1}{8}| \xi_{(a_1,a_2,a_3,a_4)}^{(+)}\rangle  \langle \xi_{(a_1,a_2,a_3,a_4)}^{(+)}|\\
+(-1)^n\left( \frac{1}{8}|\xi_{(a_1,a_2,a_3,a_4)}^{(-)}\rangle\langle \xi_{(a_1,a_2,a_3,a_4)}^{(-)}|+\frac{1}{16}| \eta_{b_1-b_2}+\eta'_{b_1-b_2} \rangle\langle \eta_{b_1-b_2}+\eta'_{b_1-b_2}|  \right) \\
+e^{in \theta} \frac{3}{32}|\varphi_+\rangle\langle \varphi_+|
+e^{-in \theta} \frac{3}{32}|\varphi_-\rangle\langle \varphi_-| .
\end{multline}
Finally, for example, if the initial state is $\varphi_0=\delta_{b_1}$, then, the survival probability can be computed by 
    \begin{align*} 
    \gamma &=|| \Pi_{\mathcal{T}}\varphi_0 ||^2+|| \Pi_{\mathcal{K}}\varphi_0 ||^2+|| \Pi_{\mathcal{C}}\varphi_0 ||^2  \\
    &= \frac{1}{{16}}|\langle \eta_{b_1-b_2}+\eta_{b_1-b_2}', \varphi_0\rangle|^2
    +{\frac{3}{32}|\langle \varphi_+,\varphi_0\rangle|^2}+{\frac{3}{32}|\langle \varphi_-,\varphi_0\rangle|^2} \\
    &= \frac{1}{16}|2|^2+\frac{3}{32}|1-\lambda_+|^2+\frac{3}{32}|1-\lambda_-|^2 \\
    &= 1/2. 
    \end{align*}
The second equality derives from the fact that the orthonormalized eigenvectors in the centered generalized eigenspace which have an overlap with the self-loop $b_1$ are given by ${(1/4)}(\eta_{b_1-b_2}+\eta_{b_1-b_2}')$ and {$\sqrt{3/32}\;\varphi_\pm$}. 
\begin{figure}[htbp]
    \centering
    \includegraphics[width=10.0cm]{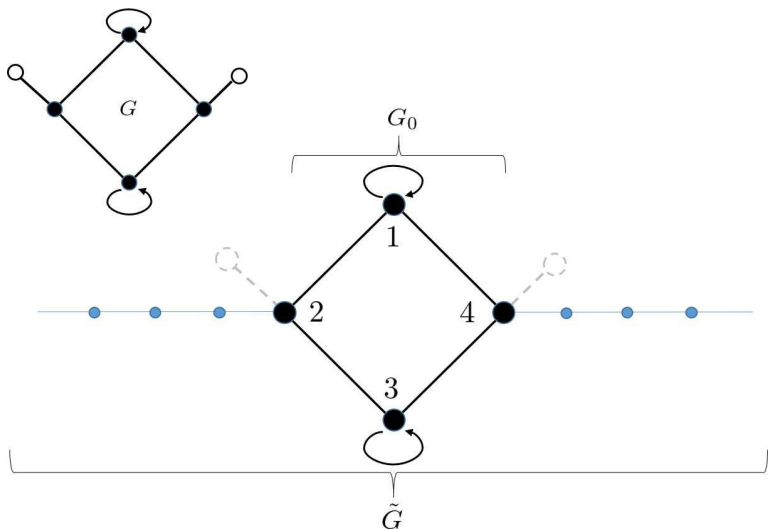}
    \caption{{\bf The setting of graphs:} The original graph $G$ is depicted at the left corner. The sinks $V_s$ are the white vertices. The subgraph $G_0$ of $G$ is the black colored graph at the center. The set of boundary vertices $\delta V$ is $\{2,4\}$. The semi-infinite graph $\tilde{G}$ is constructed by connecting the infinite length path to each boundary vertex of $G_0$. }
    \label{Fig:2}
\end{figure}
\begin{figure}[htbp]
    \centering
    \includegraphics[width=12.0cm]{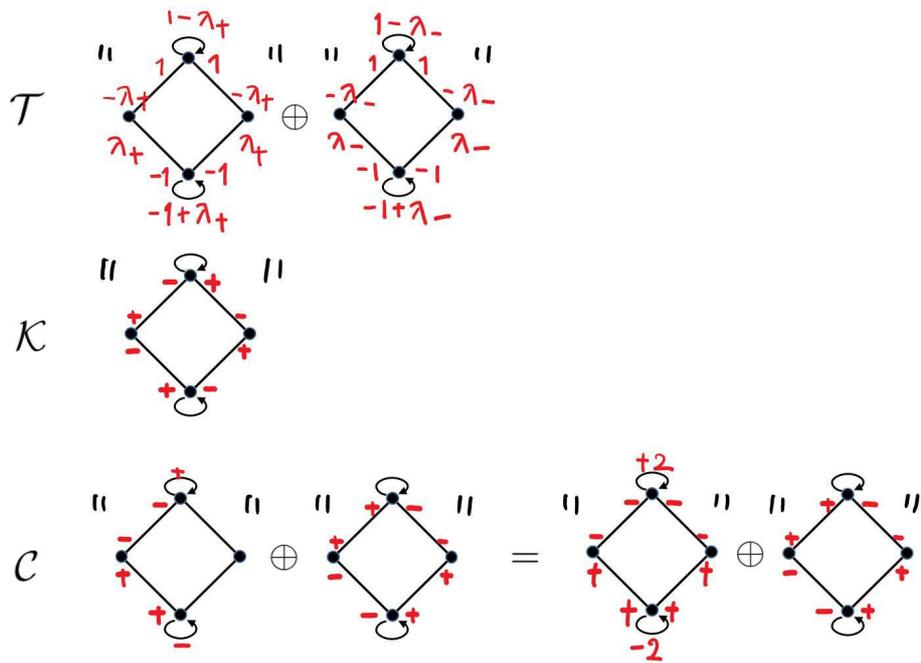}
    \caption{{\bf The centered eigenspace of the example:} The centered eigenspace to which Grover walk with sinks asymptotically belongs in this example is $\mathcal{T}\oplus \mathcal{K}\oplus \mathcal{C}$. Each weighted sub-digraph represents a function in $\mathbb{C}^{A_0}$; the complex value at each arc is the returned value of the function. Each eigenspace; $\mathcal{T}$, $\mathcal{K}$ and $\mathcal{C}$, is spanned by the functions represented by these weighted sub-digraphs. }
    \label{Fig:3}
\end{figure}
\section{Relation between Grover walk with sinks and Grover walk with tails}

\subsection{Grover walk on graphs with tails}
Let $G=(V,A)$ be a finite and connected graph with the set of sinks $V_s\subset V$. 
We introduce the infinite graph $\tilde{G}=(\tilde{V},\tilde{A})$ by adding the semi-infinite paths to each vertex of $\delta V=\{v_1,\dots,v_r\}$, that is, 
\begin{align*}
  \tilde{V} &= V(G)\setminus V_s \cup (\cup_{j=1}^rV(\mathbb{P}_j)), \\
  \tilde{A} &= \cup_{j=1}^rA(\mathbb{P}_j) \;\cup\;\left(A \setminus \{a\in A \;|\; t(a)\in V_s \;or\;o(a)\in V_s\}  \right).  
\end{align*} 
Here $\mathbb{P}_{i}$'s is the semi-infinite paths named the tail whose origin vertex is identified with $v_i$ ($i=1,\dots,r$). 
See Fig.~\ref{Fig:2}. 
Recall that $G_0=(V_0,A_0)$ is the subgraph of $G$ eliminating the sinks $V_s$. 
Recall also that $\chi_S: \mathbb{C}^{A}\to \mathbb{C}^{A_0}$ is  
    \[ (\chi_S\phi)(a)=\phi(a) \]
for all $a\in A_0$. 
In the same way, we newly introduce $\chi_T: \mathbb{C}^{\tilde{A}}\to \mathbb{C}^{A_0}$ by 
    \[ (\chi_T\phi)(a)=\phi(a) \]
for all $a\in A_0$. The adjoint $\chi_T^*: \mathbb{C}^{A_0}\to \mathbb{C}^{\tilde{A}}$ is 
    \[ (\chi_T^*f)(a)=\begin{cases} f(a) & \text{: $a\in A_0$,}\\ 0 & \text{: otherwise.} \end{cases} \]
The following theorem was proven in \cite{HS}.
\begin{theorem}[\cite{HS}]\label{thm:HS}
Let $\tilde{G}=(\tilde{V},\tilde{A})$ be the graph with infinite tails $\{\mathbb{P}_j\}_{j=1}^r$ induced by $G_0$ and its boundaries $\delta V_0$. 
Assume the initial state $\psi_0$ is 
    \[  \psi_0(a)=
    \begin{cases} 
    \alpha_1 & \text{: $a\in A(\mathbb{P}_1)$, $\dist(o(a),v_1)>\dist(t(a),v_1)$,} \\
    $\vdots$ \\
    \alpha_r & \text{: $a\in A(\mathbb{P}_r)$, $\dist(o(a),v_r)>\dist(t(a),v_r)$,} \\
    0 & \text{: otherwise.}
    \end{cases} \]
Then $\lim_{n\to\infty}\psi_n(a)=:\psi_\infty(a)$ exists and $\psi_\infty(a)$ is expressed by
    \[ \psi_\infty(a) = \frac{\alpha_1+\cdots+\alpha_r}{r}+\mathrm{j}(a). \]
Here $\mathrm{j}(\cdot)$ is the electric current flow on the electric circuit assigned the resistance value $1$ at each edge, that is, $\mathrm{j}(\cdot)$ satisfies the following properties:
    \begin{align*}
        d_1\;\mathrm{j} &= 0,\; \mathrm{j}(\overline{a})=-\mathrm{j}(a)\;\;(\mathrm{Kirchhoff's\; current\; law})\\
        \partial_2^*\;\mathrm{j} &= 0 \;\;(\mathrm{Kirchhoff's\; voltage\; law})
    \end{align*}
with the boundary conditions 
    \begin{equation}\label{eq:currentboundary} \mathrm{j}(e_i) = \alpha_i-\frac{\alpha_1+\cdots+\alpha_r}{r} \end{equation}
for any $e_i$ ($i=1,\dots,r$) such that $t(e_i)=v_j$ and $o(e_i)\in V(\mathbb{P}_i)$. 
\end{theorem}
\begin{remark}
The stationary state $\psi_\infty$ satisfies the equation 
    \[ \psi_\infty(a)=(U_{\tilde{G}}\psi_\infty)(a) \]
for any $a\in A$, and $\psi_\infty \in \ell^\infty$, however $||\psi_\infty||_{\tilde{A}}=\infty$. 
\end{remark}
\begin{remark}
The function $\xi_c^{(+)}=(1-S)\partial_2\delta_c^{(\Gamma)}$  also satisfies \[ \chi_T^{*}\xi_c^{(+)}(a)=(U_{\tilde{G}}\chi_T^{*}\xi_c^{(+)})(a) \] and Kirchhoff's current and voltage laws if the internal graph $G_0$ is not a tree, while it does not satisfy the boundary condition~(\ref{eq:currentboundary}) because the support of this function $\chi_T^{*}\xi_c^{(+)}$ has no overlaps to the tails but is included in the fundamental cycle $c$ in the internal graph $G_0$. 
\end{remark}
\subsection{Relation between Grover walk with sinks and Grover walk with tails}
Let us consider the Grover walk on $G$ with sinks $V_s$ and with the initial state $\psi_0^{(S)}\in \mathbb{C}^{A}$. 
We describe $U_G$ as the time evolution operator of Grover walk on $G$. 
The $n$-th iteration of this walk following (\ref{eq:TEsink}) is denoted by $\psi_n^{(S)}$. 
Let us also consider the Grover walk on $\tilde{G}$ with the tails and with the ``same" initial state  
\[\psi_0^{(T)}(a)= \begin{cases} \psi_0^{(S)}(a) & \text{: $a\in A_0$,}\\ 0 & \text{: otherwise.}\end{cases}\] 
Note that the initial state $\psi_0^{(S)}$ is different from the one in the setting of Theorem~\ref{thm:HS}.  
Putting the time evolution operator on $\tilde{G}$ by $U_{\tilde{G}}$, we denote the $n$-th iteration of this walk by $\psi_n^{(T)}=U_{\tilde{G}}\psi_{n-1}^{(T)}$. 
Then we obtain a simple but important relation between QW with sinks and QW with tails. 
\begin{lemma}
Let the setting of the QW with sinks and QW with tails be as the above.
Then for any time step $n$, we have 
\[ \chi_S\psi_n^{(S)}=\chi_T\psi_n^{(T)}.  \]
\end{lemma}
\begin{proof}
The initial state of $\chi_S\psi_0^{(S)}$ coincides with $\chi_T\psi_0^{(T)}$ because of the setting.  
Note that $\chi_J^*\chi_J$ is the projection operator onto $\mathbb{C}^{A_0}$ while $\chi_J\chi_J^*$ is the identity operator on  $\mathbb{C}^{A_0}$ $(J\in\{S,T\})$.  
Since $\psi_n^{(S)}(a)=0$ for any $a\in V_s$, we have 
    \[ (1-\chi_S^*\chi_S)\psi_n^{(S)}=0 \]
for any $n\in \mathbb{N}$. 
Then putting $\chi_S\psi_n^{(S)}=:\phi_n^{(S)}$ and $\chi_S U_G \chi_S^*=: E$ we have 
\begin{align*}
    \phi_{n}^{(S)} &=\chi_S\psi_n^{(S)} = \chi_S U_G\psi_{n-1}^{(S)} \\ 
       & = \chi_S U_G (\chi_S^*\chi_S + (1-\chi_S^*\chi_S)) \psi_{n-1}^{(S)} \\
        &= E \phi_{n-1}^{(S)} + (\chi_S U_G (1-\chi_S^*\chi_S))\psi_{n-1}^{(S)} \\
        &= E\phi_{n-1}^{(S)}.
\end{align*}
It is easy to see that $E=\chi_S U_G \chi_S^*=\chi_T U_{\tilde{G}} \chi_T^*$. 
Since the support of the initial state is included in the internal graph, the inflow never come into the internal graph from the tail for any time $n$, which implies 
\[(\chi_T U_{\tilde{G}} (1-\chi_T^*\chi_T))\psi_{n}^{(T)}=0. \]
It holds that $E=\chi_S \tilde{U}\chi_S^*=\chi_T U_{\tilde{G}}\chi_T^*$. 
Then putting $\phi_n^{(T)}:=\chi_T\psi_n^{(T)}$, 
in the same way as $\psi_n^{(S)}$, we have 
\begin{align*}
    \phi_{n}^{(T)}
       & = \chi_T U_{\tilde{G}} (\chi_T^*\chi_T + (1-\chi_T^*\chi_T)) \psi_{n-1}^{(T)} \\
       &= E\phi_{n-1}^{(T)}. 
\end{align*}
Therefore $\chi_S\psi_n^{(S)}$ and $\chi_T\psi_n^{(T)}$ follow the same recurrence and have the same initial state which means $\chi_S\psi_n^{(S)}=\chi_T\psi_n^{(T)}$ for any $n\in \mathbb{N}$. 
\end{proof}
\begin{corollary}
Let the initial state for the Grover walk with sinks
be $\phi_0$ with $\supp(\phi_0)\subset A_0$. 
The survival probability $\gamma$ can be expressed by 
    \[ \gamma= ||\phi_0||_{A}^2-\sum_{n=0}^\infty \tau_n, \]
where $\tau_n$ is the outflow of the QW with tails from the internal graph $G_0$, i.e., 
    \[ \tau_n= \sum_{o(a)\in\delta V, \; t(a)\notin A_0} |(U_{\tilde{G}}\chi_T^{*}\phi_{n-1}^{(T)})(a)|^2 \]
\end{corollary}
\begin{remark}
The time evolution for $\phi_n^{(T)}$ is given by 
    \[ \phi_{n}^{(T)}=E\phi_{n-1}^{(T)}+\rho, \]
where $\rho= \chi_T U_{\tilde{G}}\psi_0^{(T)}$. In this case, the inflow is $\rho=0$. On the other hand, in the setting of Theorem~\ref{thm:HS}, $\rho$ is given by a nonzero constant vector.
\end{remark}
Let us now consider a QW with tails with a general initial state $\Psi_0\in \mathbb{C}^{\tilde{A}}$ on $\tilde{G}$. We denote $\nu=\chi_T\Psi_0$ and 
$\rho=\chi_T U_{\tilde{G}}(1-\chi^*\chi)\Psi_0$. 
We summarize the relation between a QW with sinks and a QW for the setting of Theorem~\ref{thm:HS} in the table from the view point of a QW with tails. 
\begin{center}
\begin{tabular}{l|c|c|c}
     &  $\rho$ & $\nu$ & state in $G_0$\\ \hline
QW with tails in the setting of Thm~\ref{thm:HS}\cite{HS}     & $\neq 0$ & $=0$ & $\in \mathcal{H}_s$ (for any $n$)\\
QW with sinks & $=0$ & $\neq 0$ & $\in \mathcal{H}_c$ (asymptotically)
\end{tabular}
\end{center}
\section{Centered generalized eigenspace of $E$ for the Grover walk case}
\subsection{The stationary states from the view point of the centered generalized eigenspace}
From the above discussion, we see the importance of the spectral decomposition 
$$
E=\chi_S U_G\chi_S^*=\chi_T U_{\tilde{G}}\chi_T^* ,
$$
to obtain both limit behaviors. 
The operator $E$ is no longer a unitary operator and, moreover, it is not ensured that it is diagonalizable. 
The centered generalized eigenspace of $E$ is defined by 
    \[ \mathcal{H}_c:= \{\psi\in \mathbb{C}^{A_0} \;|\; \exists\; m\geq 1 \mathrm{\;and\;} \exists\;|\lambda|=1 \mathrm{\;such\;that\;}(E^m-\lambda) \psi=0 \} \]
Let $\mathcal{H}_s$ be defined by 
    \[ \mathbb{C}^{A_0}=\mathcal{H}_c \oplus \mathcal{H}_s. \]
Here ``$\oplus$" means $\mathcal{H}_c$ and $\mathcal{H}_s$ are complementary spaces, that is, if $u_c+u_v=0$ for some $u_c\in \mathcal{H}_c$ and $u_v\in \mathcal{H}_s$, then $u_c$ and $u_v$ must be $u_c=u_v=0$. 
Note that since $E$ is not a normal operator on a vector space $\mathcal{H}_c\oplus \mathcal{H}_s$, it seems that in general $\langle u_c,u_v\rangle\neq 0$ for $u\in \mathcal{H}_c$ and $\mathcal{H}_s\in N$. 
However, we can see some important properties of the spectrum of $E$ in the following proposition.
\begin{proposition}[\cite{HS}]\label{prop:center}
\noindent
\begin{enumerate}
\item For any $\lambda\in \mathrm{Spec}(E)$, it holds that $|\lambda|\leq 1$, i.e., 
\[ \mathcal{H}_s= \{ \psi\;|\; \exists\;m\in\mathbb{N},\;\exists\;|\lambda|<1,\;(U-\lambda)^m \psi)=0 \}. \]
\item Let $P_c$ be the projection operator on $\mathcal{H}_c$ along with $\mathcal{H}_s$; that is, $P_cE=EP_c$ and $P_c^2=P_c$. Then $P_c$ is the orthogonal projection onto $\mathcal{H}_c$, i.e., $P_c=P_c^*$.  
\item The operator $E$ acts as a unitary operator on $\mathcal{H}_c$, that is, 
$\mathcal{H}_c=\oplus_{|\lambda|=1}\ker(\lambda-E)$ and $U_G\chi_S^*\varphi=\lambda \chi_S^*\varphi$ for any $\varphi\in \ker (\lambda-E)$ with $|\lambda|=1$.
\end{enumerate}
\end{proposition}
We call $\mathcal{H}_c$ and $\mathcal{H}_s$ the {\it centered eingenspace} and the {\it stable eigenspace}~\cite{R}, respectively. 
\begin{corollary}\label{cor:orthogonal}
For any $\psi\in \mathcal{H}_s$ and $\phi\in \mathcal{H}_c$, it holds that 
$\langle \psi,\phi \rangle=0$. 
\end{corollary}
Now let us see the stationary states from the view point of the {\it orthogonal} decomposition of $\mathcal{H}_c\oplus \mathcal{H}_s$. 
\begin{proposition}\label{prop:stationary}
\noindent
\begin{enumerate}
\item The state $\chi_T\psi_n$ in Theorem~\ref{thm:HS} belongs to $\mathcal{H}_s$ for any time step $n\in \mathbb{N}$. 
\item The state of QW with sinks; $\chi_S\phi_n$, asymptotically belongs to $\mathcal{H}_c$ in the long time limit $n$.  
\end{enumerate}
\end{proposition}
\begin{proof}
The inflow $\rho=\chi^*U\psi_0$ is orthogonal to $\mathcal{H}_c$ by a direct consequence of Lemma~3.5 in \cite{HS}, which implies $E^n\rho\in \mathcal{H}_s$ for any $n\in \mathbb{N}$ by Proposition~\ref{prop:center}.  
Since the stationary state of part 1 is described by the limit of the following recurrence  
$$
\chi_T\psi_{n}=E\chi_T\psi_{n-1}+\rho, \quad  \chi_T\psi_0=0 ,
$$
we obtain the conclusion of part 1. 
On the other hand, let us consider the proof of part 2 in the following. The time evolution in $G_0$ obeys $\chi_S\phi_n=E\chi_{S}\phi_{n-1}$. 
The overlap of $\chi_S\phi_n$ to the space $\mathcal{H}_s$ decreases faster than polynomial times  because all the absolute value of the generalized eigenvalues of $\mathcal{H}_s$ are strictly less than $1$. (See Proposition~\ref{prop:convergence} for more detailed order of the convergence.) Then only the contribution of the centered eigenspace, whose eigenvalues lie on the unit circle in the complex plain, remains in the long time limit. 
\end{proof}
Let $W=P_cE=EP_c=P_cEP_c$ be the operator restricted to the centered eigenspace $\mathcal{H}_c$. Then we have 
    \[ \lim_{n\to\infty}\left| \chi_S\phi_n(a)-W^n\chi_S\phi_0(a)\right|=0 \]
for any $a\in A_0$ uniformly by Proposition~\ref{prop:stationary}. 
This means that in the long time limit, the time evolution is reduced to $W$ which is a unitary operator on $\mathcal{H}_c$.  
\begin{proposition}\label{prop:convergence}
The survival probability is re-expressed by 
    \[ \gamma=  ||P_c\chi_S\phi_0||^2. \]
The convergence speed\footnote{$f(n)=O(g(n))$ means $\lim_{n\to\infty} |f(n)/g(n)|<\infty$ if the limit exists. } is estimated by $O(n^\kappa r_{max}^n)$, where $\kappa=\dim \mathcal{H}_s$, $r_{max}=\max\{|\lambda| \;;\; \lambda\in \mathrm{Spec}(E),\;|\lambda|<1 \}$.
\end{proposition}
\begin{proof}
Putting $E(1-P_c)=W'$, we have 
$$
W + W' = E, \quad W W' = 0,
$$ 
by Proposition~\ref{prop:center} (2).
Note that the operator $E^n$ is similar to 
    \[ \bigoplus_{\lambda\in \mathrm{Spec}(E)}J^n(\lambda;k_\lambda) \]
with some natural numbers $k_\lambda$'s. Here $J(\lambda;k)$ is the $k$-dimensional matrix by 
\[ J(\lambda;k)=\begin{bmatrix} 
\lambda & 1       &        &        & \\
        & \lambda & 1      &        & \\
        &         & \ddots & \ddots & \\
        &         &        & \ddots & 1 \\
        &         &        &        &\lambda
\end{bmatrix}.\]
We obtain that the survival probability at each time $n$ is described by    
    \begin{align*}
        \gamma_n 
        &= || U_G\chi_S^* E^{n-1} \chi_S\phi_0 ||^2 \\
        & = || U_G\chi_S^* (W^{n-1}+{W'}^{n-1}) \chi_S\phi_0 ||^2 \\
        &= || (W^{n-1}+{W'}^{n-1}) \chi_S\phi_0 ||^2 \\
        &= || W^{n-1} \chi_S\phi_0 ||^2  + || {W'}^{n-1} \chi_S\phi_0 ||^2. 
    \end{align*}
In the third equality we have used the fact that $U_G$ is unitary, the last equality follows from Corollary~\ref{cor:orthogonal}.
The second term decreases to zero by  Proposition~\ref{prop:center} (2) with the convergence speed at least $O(n^{\kappa}r_{max}^n )$ because the Jordan matrix $J(\lambda;k)$ can be estimated by  $J(\lambda;k)^n=O(n^k|\lambda|^n)$.
Hence, we find for $\gamma_n$    
    \begin{align*}
     \gamma_n   &= || W^{n-1}\chi_S\phi_0  ||^2 + O(n^\kappa r_{max}^n) \;(n>>1)\\
        &= || W^{n-1}P_c\chi_S\phi_0  ||^2 + O(n^\kappa r_{max}^n) \\
        &= || P_c\chi_S\phi_0  ||^2 + O(n^\kappa r_{max}^n), 
    \end{align*}
where in the second equality we have used that $W = W P_c$ and the last equality follows from Proposition~\ref{prop:center} (3).
\end{proof}

Therefore, the characterization of $\mathcal{H}_c$ is important to obtain the asymptotic behavior of $\phi_n$.  
\subsection{Characterization of centered generalized eigenspace by graph notations}
The centered generalized eigenspace of $E$ can be rewritten by using the boundary operator $d_1$ and the self-adjoint operator $T=d_1Sd^*_1$ as follows. 
\begin{lemma}[\cite{HS}]\label{lem:HS}
Assume $\lambda\in \mathrm{Spec}(E)$ with $|\lambda|=1$. Then we have
\begin{enumerate}
    \item $\lambda=\pm 1$ if and only if  $\ker(\lambda-E)=\ker(-\lambda-S)\cap \ker d_1$;
    \item $\lambda \neq \pm 1$ if and only if  $\supp (g)\subset V_0\setminus \delta V_0$ for any $g\in \ker((\lambda+\lambda^{-1})/2-T)\neq 0$. 
\end{enumerate}
\end{lemma}
In the following, we consider the characterization of $\ker(\pm 1 -E)$ using some walks on graph $G_0$ up to the situations of the graph; cases (A)--(D). 
First we prepare the following notations.
For each support edge $e\in E_0$, there are two arcs $a$ and $\overline{a}$ such that $|a|=|\overline{a}|$. 
Let us choose one of the arcs from each $e\in E_0$ and denote $A_+$ as the set of selected arcs. 
Then $|A_+|=|E_0|$ and $a\in A_+$ if and only if $\overline{a}\notin A_+$ holds. 
We set $A_{rep}=A_{0,\sigma} \cup A_+$. 
Let us introduce the map $\iota: \mathbb{C}^{A_0}\to \mathbb{C}^{A_{rep}}$ defined by $(\iota \psi)(a)=\psi(a)$ for any $\psi\in \mathbb{C}^{A_0}$ and $a\in A_{rep}$. 

Let us define the boundary operator $\partial_+ : \mathbb{C}^{A_{rep}}\to \mathbb{C}^{V_0}$ by 
    \[ (\partial_+ \varphi)(u) = \sum_{ t(a)=u \mathrm{\;in\;} A_+}\varphi(a)-\sum_{o(a)=u\mathrm{\;in\;} A_+}\varphi(a) \]
for any $\varphi\in \mathbb{C}^{A_{rep}}$ and $u\in V_0$. 
On the other hand, let us also define the boundary operator $\partial_- : \mathbb{C}^{A_{rep}}\to \mathbb{C}^{V_0}$ by 
    \[ (\partial_- \varphi)(u) 
    =\begin{cases}
    \sum\limits_{t(a)=u}\varphi(a)+\sum\limits_{o(a)=u}\varphi(a) & \text{: $u$ has no selfloop,} \\
    \sum\limits_{t(a)=u}\varphi(a)+\sum\limits_{o(a)=u}\varphi(a) -\varphi(a_s)& \text{: $u$ has a selfloop $a_s$,}
    \end{cases} \]
for any $\varphi\in \mathbb{C}^{A_{rep}}$ and $u\in V_0$. 
We obtain the following lemma. 
\begin{lemma}\label{lem:dimension}
Let $G_0=(V_0,A_0)$ be a graph with self-loops. 
We set $E_0$ as the set of support edges of $A_0\setminus A_{0,\sigma}$ such that $E_0=\{ |a| \;|\; a\in A_0\setminus A_{0,\sigma} \}$. 
Then we have 
\[ \dim [\ker(1-E)] = |E_0|-|V_0|+1, \]
\[ \dim [\ker(1+E)] = 
\begin{cases}
|E_0|-|V_0|+1 & \text{: Case A,} \\
|E_0|-|V_0| & \text{: Case B,} \\
|E_0|-|V_0|+|A_{0,\sigma}| & \text{: Cases C and D,}
\end{cases} \]
\end{lemma}
\begin{proof}
Note that if $\psi\in \ker(1+S)$, then $\psi(\overline{a})=-\psi(a)$ for any $a\in A_+$ and if $\psi \in \ker(d)$, then $\sum_{t(a)=u}\psi(a)=0$ for any $u\in V_0$. 
Remark that since $(S\psi)(a_s)=\psi(a_s)$ for any $a_s\in A_{0,\sigma}$, we have $\psi(a_s)=0$ if $\psi\in \ker(1+S)$. 
Therefore if $\psi\in \ker(1+S) \cap \ker(d)$, then 
\[ \sum_{t(a)=u\mathrm{\;in\;}A_{+}} (\iota \psi)(a)-\sum_{o(a)=u\mathrm{\;in\;}A_{+}} (\iota \psi)(a)=(\partial_+\iota \psi)(u)=0  \]
holds. 
Then $\ker(1+S)\cap \ker d$ is isomorphic to $\{ \varphi\in \ker \partial_+ \;|\; \supp(\varphi)\subset A_+ \}$. 
Let us consider $\ker \partial_+$. 
By the definition of $\partial_+$, we have  $\partial_+\delta_a^{(A_{rep})}=0$ for any $a\in A_s$.
Hence, we should eliminate the subspace of $\ker\partial_+$ induced by the self-loops. The dimension of this subspace is $|A_{0,\sigma}|$.
The adjoint operator $\partial^*_+:\mathbb{C}^{V_0}\to \mathbb{C}^{A_+}$ of  $\partial_+$ is described by 
    \[ (\partial^*_+ f)(a) = 
    f(t(a))-f(o(a)) , \]
for any $f\in\mathbb{C}^{V_0}$ and $a\in A_{rep}$. 
If $\partial_+^*f=0$ holds, then $f(t(a))=f(o(a))$ for any $a\in A_+$. This means $f(u)=c$ for any $u\in V_0$  with some non-zero constant $c$. Thus $\dim \ker(\partial_+^*)=1$.  Therefore, the fundamental theorem of linear algebra\footnote{for a linear map $g: X\to Y$, $\dim\ker g=\dim X-\dim Y+\dim \ker g^*$. } implies 
    \begin{align*} 
    \dim\ker(1+S)\cap \ker d &= \dim \ker (\partial_+)-|A_{0,\sigma}|\\
    &= (|A_{rep}|-|V_0|+1)-|A_{0,\sigma}| \\
    &= |E_0|-|V_0|+1.
    \end{align*}

Next, let us consider $\dim(\ker(1-S)\cap\ker d_1)$. 
Note that if $\psi\in \ker(1-S)$, then $\psi(\overline{a})=\psi_(a)$. 
Assume that $\psi\in \ker(1-S)\cap \ker(d_1)$, then 
\[ \sum_{t(a)=u}(\iota\psi)(a)=0 \mathrm{\;for\; any\;} u\in V_0 ,\]
which is equivalent to 
\[ \partial_-\iota\psi=0. \]
The adjoint of $\partial_-$ is described by 
    \[ (\partial_-^*f)(a)=\begin{cases} f(t(a))+f(o(a)) & \text{: $a\in A_+$,} \\ f(t(a)) & \text{: $a\in A_{0,\sigma}$.} \end{cases} \]
Let us consider $f\in \ker(\partial_-^*)$ in the cases for both  $A_{0,\sigma}=\emptyset$ and $A_{0,\sigma}\neq \emptyset$. \\
\noindent {\bf $A_{0,\sigma}=\emptyset$ case}: \\
If $G_0$ is a bipartite graph, then we can decompose the vertex set $V$ into $X \cup Y$, where every edge connects a vertex in $X$ to one in $Y$. 
Then $f(x)=k$ for any $x\in X$ and $f(y)=-k$ for any $y\in Y$ with some nonzero constant $k$. Hence, $\dim \ker(\partial^*)=1$ if $A_{0,\sigma}=\emptyset$ and $G_0$ is bipartite.
On the other hand, if $G_0$ is non-bipartite, then there must exist an odd length fundamental cycle $c=(a_0,a_1,\dots,a_{2m})$. 
We have that
$$
f(o(a_1))=-f(o(a_2))=f(o(a_3))=\cdots=-f(o(a_{2r}))=f(o(a_0))=-f(o(a_1)).
$$
Then $f(u)=0$ for any $u\in V(c)$. Since $G_0$ is connected, the value $0$ is inherited to the other vertices by $f(t(a))=-f(o(a))$. After all, we have $f=0$, which implies $\ker(\partial_-^*)=0$ if $A_{0,\sigma}=\emptyset$ and $G_0$ is non-bipartite.  \\
\noindent{\bf $A_{0,\sigma}\neq \emptyset$ case}: \\
Since $(\partial_-^{*}f)(a)=f(t(a))=0$ if $a\in A_{0,\sigma}$, then $f$ takes the value $0$ at the other vertices since $f(t(a))=-f(o(a))$ for any $a\in A_+$, which implies $\ker(\partial_+^*)=0$ if $A_{0,\sigma}\neq \emptyset$.  \\ 
After all by the fundamental theorem of the linear algebra, 
\begin{align*} 
\dim \ker \partial_- 
&= |A_{rep}|-|V_0|+\begin{cases} 1 & \text{: $A_s=\emptyset$, $G_0$ is bipartite.}\\
0 & \text{: otherwise.}\end{cases}
\end{align*}
Noting that $|A_{rep}|=|E_0|+|A_{0,\sigma}|$, we obtain the desired conclusion.
\end{proof}
In the following, let us find linearly independent eigenfunctions of $\ker(\pm 1-E)$ using some concepts from graph theory. 
A walk in $G_0$; $p=(a_0,a_1,\dots,a_r)$ is a sequence of arcs with $t(a_j)=o(a_{j+1})$ ($j=0,1,\dots,r-1$), which may have the same arcs in $p$. We set $\{ a_0,a_1,\dots,a_r \}=:A(p)$, and similarly $\overline{A}(p) = \{\overline{a}_0,\dots,\overline{a}_{r}\}$ as a {\it multi} set. 
We describe $\tilde{\xi}_p^{(\pm )}:\{a_0,\dots,a_{r}\}\cup \{\overline{a}_0,\dots,\overline{a}_{r}\}\to \{\pm 1\}$ by 
    \[ \tilde{\xi}_p^{(+)}(a)= \begin{cases} 1 & \text{: $a\in A(p)$,}\\ -1 & \text{: $\overline{a}\in A(p)$,} \\ \end{cases}  \]
    \[ \tilde{\xi}_p^{(-)}(a)= \begin{cases} 1 & \text{: $|a|\in \{|a_j| \;|\; j \mathrm{\;is\;even}\}$,}\\ -1 & \text{: $|a|\in \{|a_j| \;|\; j \mathrm{\;is\;odd} \}$.} \end{cases}  \]
Then we set the functions $\xi_p^{(\pm)}\in \mathbb{C}^A$ by 
    \begin{equation}\label{eq:defxi}  
    \xi_p^{(\pm)}(a)=
    \begin{cases} 
    \sum_{b:\;a=b}\tilde{\xi}_p^{(\pm)}(b) & \text{: $a\in A(p)\cup  \overline{A}(p)$,}\\
    0 & \text{: otherwise.}
    \end{cases}
    \end{equation}
Now we are ready to show the following proposition for $\ker(1-E)$. 
\begin{proposition}\label{prop:h+}
Let $\xi_c^{(+)}$ be defined as (\ref{eq:defxi}). Then we have 
\[ \ker(1-E) = \mathrm{span}\{ \xi_c^{(+)} \;|\; c\in \Gamma \}. \]
\end{proposition}
\begin{proof}
By the definition of $\xi_c^{(+)}$, we have 
$\xi_c^{(+)}\in \ker d_1 \cap \ker(1-S)$ which implies $\xi_c^{(+)}\in \ker(1-E)$ by Lemma~\ref{lem:HS}. 
We show the linear independence of $\{\xi_c^{(+)}\}_{c\in \Gamma}$. 
Let us set $\Gamma=\{c_1,\dots,c_r\}$ and $\xi_j:=\xi_{c_j}^{(+)}$ ($j=1,\dots,r$) induced by the spanning tree $\mathbb{T}\subset G$. Assume that
\[ \beta_1\xi_1+\cdots+\beta_r \xi_r=0. \]
Put $a_{r}\in A_0(c_{r})\cap (A_0\setminus A(\mathbb{T}))$. 
From the definition of the fundamental cycle, we have 
    \[ \beta_1 \xi_1(a_r)+\cdots+\beta_r \xi_r(a_r)  = \beta_r= 0. \]
In the same way, let $a_{r-1}\in A(c_{r-1})\cap (A_0\setminus A(\mathbb{T}))$, then 
    \[ \beta_1 \xi_1(a_r)+\cdots+ \beta_{r-1} \xi_{r-1}(a_{r-1}) = \beta_{r-1}= 0. \]
Then using it recursively, we obtain $\beta_1=\cdots=\beta_r=0$ which means $\xi_j$'s are linearly independent. 
Then $\dim(\mathcal{K})=|\Gamma|=|E_0|-|V_0|+1$. 
By Lemma~\ref{lem:dimension}, we reached to the conclusion. 
\end{proof}
Define $\Gamma_o,\Gamma_e\subset \Gamma$ as the set of odd and even length fundamental cycles. In the following, to obtain a characterization of $\ker(1+E)=\ker(1-S)\cap \ker (d_1)$, we construct the function $\eta_{x,y}\in \ker(1-S)\cap \ker (d_1)$ which is determined by $x,y\in A_{0,\sigma} \cup \Gamma_o$. 
The main idea to construct such a function is as follows. 
By the definition of $\xi^{(-)}_q$ for any walk $q$, $\xi^{(-)}_q\in \ker(1-S)$. This is equivalent to assigning the symbols ``$+$" and ``$-$" alternatively to each edge along the walk $q$. If the walk $c$ is an even length cycle, then a symbol on each edge of $c$ is different from the ones on the neighbor's edges; this means 
$$
\sum_{t(a)=u}\xi_c^{(-)}(a)=0 ,
$$
for every $u$. Then $\xi_c^{(-)}\in \ker(d_1)\cap \ker(1-S)$ holds.  On the other hand, if the walk $c=(b_1,\dots,b_r)$ is an odd length cycle, then a ``frustration" appears at $u:=o(b_1)$; i.e., 
$$
\sum_{t(a)=u}\xi_c^{(-)}(a)=2.
$$
To vanish this frustration, there are two ways; the first is to make a cancellation by another frustration induced by another odd cycle $c'$; the second one is to push the frustration to a self-loop. That is the reason for why the domains of $x$ and $y$ are $A_{0,\sigma} \cup \Gamma_o$. We give more precise explanations of the constructions as follows. See also Fig.~\ref{Fig:1}.\\

\noindent {\bf Construction of $\eta_{x,y}\in \mathbb{C}^{A_0}$: }\\
The function $\eta_{x,y}$ is described by $\xi^{(-)}_q$ induced by a walk depending on the indexes of $x,y$.  
In this paper, we consider four cases of the domains of $x$ and $y$; (1) $x\in \Gamma_o$, $y\in \Gamma_o$;  (2) $x\in A_{\sigma}$, $y\in A_{\sigma}$; (3) $x\in A_{\sigma}$, $y\in \Gamma_o$; (4) $x\in \Gamma_o$, $y\in A_{\sigma}$.
\begin{enumerate}
\item $x\in \Gamma_o$, $y\in \Gamma_o$ case: \\
If $G_0$ is a bipartite graph, let us fix an odd length fundamental cycle  $c_*=(a_0,\dots,a_{r-1})\in \Gamma_o$ and pick up another $c\in \Gamma_o=(b_0,\dots,b_{s-1})$. 
We set the following walk $q$ and define the function on $\mathbb{C}^{A_0}$; $\xi^{(-)}_q=:\overline{a}a_{c_*-c}$, induced by $c_*,c\in \Gamma_o$: 
\begin{enumerate}
\item $c_0 \cap c\neq \emptyset$ case: 
We set $q$ as the shortest closed walk starting from a vertex $u_0\in V(c_0)\cap V(c)$ and visiting all the vertices of $V(c_0)$ and $V(c)$; that is, $q=(a_i,\dots,a_{i+r},b_{j},\dots,b_{s+j})$. Here $o(a_i)=o(b_j)=u_0$ and the suffices are modulus of $r$ and $s$. 
\item $c_0 \cap c=\emptyset$ case:
Let us fix the shortest path between $c_0$ and $c$ by $p=(p_1,\dots,p_t)$. 
Denoting the vertex in $V(c_*)$ connecting to $p$ by $u_0\in V(c_*)$, we set $q$ by the shortest closed walk $q$ starting from $u_*$ and visiting all the vertices; that is, $q=(a_{i},\dots,a_{r+i},p_0\dots,p_t,b_{j}\dots,b_{s+j},\bar{p}_t\dots,\bar{p}_1)$, where $o(a_i)=t(a_{r+i})=o(p_1)=u_0$, $t(p_t)=o(b_j)=t(b_{s+j})$. 
\end{enumerate}
Note that by the definition of the fundamental cycle, the intersection $c_0 \cap c$ is a path in the case for (1). Since $G_0$ is connected, there is a path connecting $c_*$ to $c$ and we fix such a path for every pair of $(c_*,c)$ in the case for (2). 
\item $x\in A_{\sigma}$ and $y\in A_{\sigma}$ case: \\
If the number of self-loops $|A_{\sigma}|\geq 2$, let us fix a self-loop $a_*$ from $A_{\sigma}$ and a path between $a_*$ to each $a\in A_{\sigma}\setminus\{a_*\}$. 
Let us denote the path between $a_*$ and $a$ by $p=(p_1,\dots,p_{t})$. 
Then we set the walk from $a_*$ to $a$ by $q=(a_*,p_1,\dots,p_t,a)$ and $\xi^{(-)}_q=: \eta_{a_*-a}$. 
\item $x\in A_{\sigma}$ and $y\in \Gamma_o$ case: \\
If $|A_{\sigma}|\geq 1$ and $G\setminus A_{\sigma}$ is a non-bipartite graph, let us fix a self-loop $a_*$ and pick up an odd cycle $c=(b_1,\dots,b_t)\in \Gamma_o$;  
if the self-loop $o(a_*)\in V(c)$, we set the walk starting from $a_*$ visiting all the vertices $V(c)$ and returning back to $a_*$ by $q=(a_*,b_1,\dots,b_t, a_*)$; 
while $o(a_*)\notin V(c)$, let us fix a path $p=(p_1,\dots,p_t)$ between $o(a_*)$ and $o(b_1)$ and set the walk starting from $a_*$ visiting all the vertices $V(p)\cup V(c)$ and returning back to $a_*$; $q=(a_*,p_1,\dots,p_t,b_0\dots,b_t,\bar{p}_t,\dots,\bar{p}_1,a_*)$. 
Then we set $\xi_q^{(-)}=: \eta_{a_*,c}$. 
\item $x\in \Gamma_o$ and $y\in A_{\sigma}$ case: \\
Let us fix an odd length fundamental cycle $c_*\in \Gamma_o=(b_1,\dots,b_{s-1})$ and pick up a self-loop $a\in A_{\sigma}$. Let us set a short length path $p$ between $o(a)$ and $o(b_1)$. Then we consider the same walk $q$ as in (3) and set $\xi_q^{(-)}=: \eta_{c_*,a}. $ 
\end{enumerate}
By the construction, we have $\eta_{x,y}\in \ker(1-S)\cap \ker(d_1)$. 
Using the function $\eta_{x,y}$, we obtain the following characterization of $\ker(-1-E)$. 
\begin{proposition}\label{prop:h-}
Let $\xi_c^{(-)}$ be defined by (\ref{eq:defxi}) and $\eta_{x,y}$ be the above.  
Let us fix $a_*\in A_{\sigma}$ and $c_*\in \Gamma_o$.  
Then we have 
\[ \ker(1+E) = 
\begin{cases}
\mathrm{span}\{ \xi_c^{(-)} \;|\; c\in \Gamma \} & \text{: Case (A),}\\
\mathrm{span}\{ \xi_c^{(-)} \;|\; c\in \Gamma_e \} \oplus \mathrm{span}\{ \eta_{c_*-c} \;|\; c\in \Gamma_o\setminus\{c_*\} \} & \text{: Case (B),}\\
\mathrm{span}\{ \xi_c^{(-)} \;|\; c\in \Gamma \} \oplus \mathrm{span}\{ \eta_{a_*-a} \;|\; a\in A_{0,\sigma}\setminus\{a_*\} \} & \text{: Case (C),}\\
\mathrm{span}\{ \xi_c^{(-)} \;|\; c\in \Gamma_e \} \oplus \mathrm{span}\{ \eta_{a_*-y} \;|\; y\in \Gamma_o\cup (A_{0,\sigma}\setminus\{a_*\})  \} & \text{: Case (D).}
\end{cases} \]
\end{proposition}
\begin{proof}
We put 
\begin{align}
    \mathcal{A} &:= \mathrm{span}\{ \xi_c^{(-)} \;|\; c\in \Gamma \}, \label{A}\\
    \mathcal{B} &:= \mathrm{span}\{ \xi_c^{(-)} \;|\; c\in \Gamma_e \} \oplus \mathrm{span}\{ \eta_{c_*-c} \;|\; c\in \Gamma_o\setminus\{c_*\} \}, \label{B}\\
    \mathcal{C} &:=  \mathrm{span}\{ \xi_c^{(-)} \;|\; c\in \Gamma \} \oplus \mathrm{span}\{ \eta_{a_*-a} \;|\; a\in A_{0,\sigma}\setminus\{a_*\} \}, \label{C}\\
    \mathcal{D} &:= \mathrm{span}\{ \xi_c^{(-)} \;|\; c\in \Gamma_e \} \oplus \mathrm{span}\{ \eta_{a_*-y} \;|\; y\in \Gamma_o\cup (A_{0,\sigma}\setminus\{a_*\})  \}.\label{D}
 \end{align}
See also Figure~\ref{Fig:4}. 
 From the construction of $\eta_{x,y}$ and $\xi_c^{(-)}$, the linear independence is immediately obtained. 
Let us check the dimensions for each case. \\
In Case (A); 
\[\dim(\mathcal{A})=|\Gamma|=|E_0|-|V_0|+1.\] 
In Case (B); \[\dim(\mathcal{B})=|\Gamma_e|+(|\Gamma_o|-1)=|E_0|-|V_0|.\]
In Case (C); 
\[\dim(\mathcal{C})=|\Gamma|+(|A_{0,\sigma}|-1)=|E_0|-|V_0|+|A_{0,\sigma}|.\] 
In Case (D); 
\[
\dim(\mathcal{D})=|\Gamma_e|+(|\Gamma_o|-1)+ (|A_{0,\sigma}|-1)=|E_0|-|V_0|+|A_{0,\sigma}|.
\]
By Lemma~\ref{lem:dimension}, we reached to the conclusion. 
\end{proof}
\begin{remark}
 ``$M \oplus N$" in Proposition~\ref{prop:h-} means that $M$ and $N$ are just  complementary spaces; the orthogonality is not ensured in general.  
\end{remark}
\begin{remark}
If $|\Gamma_o|=1$ in Case (B), we have $\mathcal{B}=\spann\{\xi_c^{(-)} \;|\; c\in \Gamma_e\}$. 
If $|A_{0,\sigma}|=1$ in Case (C), we have $\mathcal{C}=\spann\{\xi_c^{(-)} \;|\; c\in \Gamma\}$. 
\end{remark}
\begin{remark}
The subspace $\mathcal{D}$ can be reexpressed by 
    \[\mathcal{D}=\mathrm{span}\{ \xi_c^{(-)} \;|\; c\in \Gamma_e \} \oplus \mathrm{span}\{ \eta_{c_*-y} \;|\; y\in (\Gamma_o\setminus\{c_*\}) \cup A_{0,\sigma}  \}.\]
\end{remark}

\begin{figure}[htbp]
    \centering
    \includegraphics[width=15.0cm]{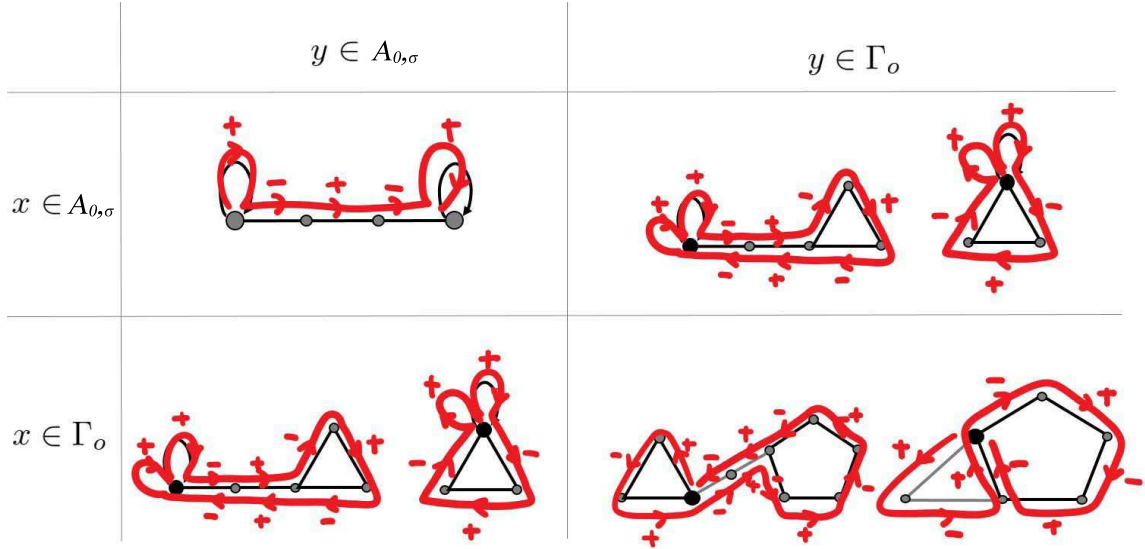}
    \caption{ {\bf Construction of eigenfunction $\eta_{x,y}\in \mathbb{C}^{A_0}$}: Each graph with signs $\pm$ represents the function $\eta_{x,y}$. The support of $\eta_{x,y}$ is included in the arcs of each graphs. The signs are the return values of this function at each arcs. The return values of the inverse arcs are the same as the original arcs. The sings are assigned alternatively along the red colored walks. At each time where the walk runs through an arc, we take the sum of the signs; e.g, in the case for $x\in A_{0,\sigma},y\in \Gamma_o$, the walk runs through the self-loop twice, then the return value at the self-loops of the function is $1+1=2$.  }
    \label{Fig:1}
\end{figure}
   
\begin{figure}[htbp]
    \centering
    \includegraphics[width=15.0cm]{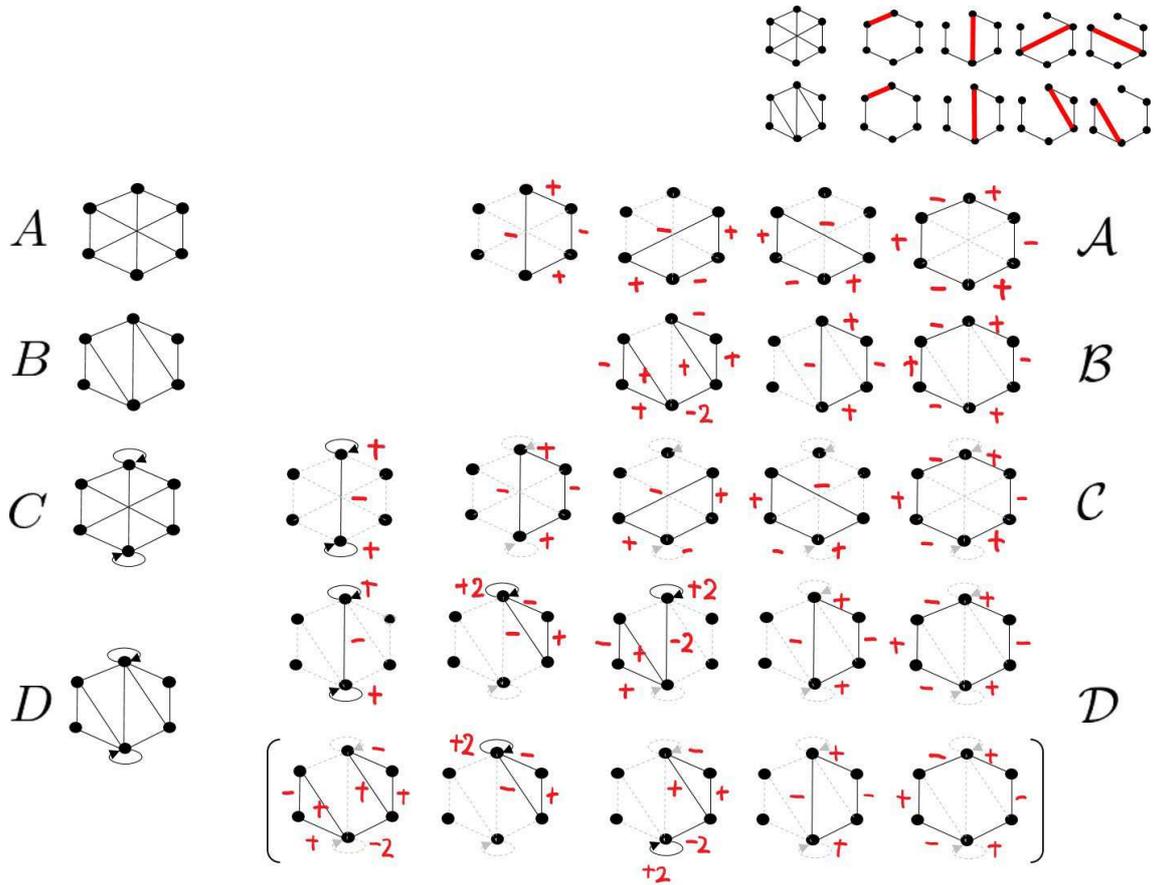}
    \caption{ {\bf Eigenspaces $\mathcal{A},\mathcal{B},\mathcal{C},\mathcal{D}$}: This figure shows examples of four graphs for the cases $A$, $B$, $C$, $D$, and its induced eigenspaces of the Grover walk $\mathcal{A}$, $\mathcal{B}$, $\mathcal{C}$, $\mathcal{D}$. The figures at the right corner are the fundamental cycles for each case. The weighted graphs represent bases of each eigenspace. The weights are the return values at each arcs of the bases, where every base takes the value $0$ at the dashed arcs. }
    \label{Fig:4}
\end{figure}
   
\section{Conclusions}   

We have investigated the Grover walk on a finite graph $G$ with sinks using its connection with the walk on the graph $G_0$ with tails. It was shown that the centered generalized eigenspace of the Grover walk with tails corresponds to the attractor space of the Grover walk with sinks, i.e., it contains all trapped states which do not contribute to the transport of the quantum walker into the sink. Consequently, the attractor space of the Grover walk with sinks can be characterized using the persistent eigenspace of the underlying random walk whose supports have no overlaps to the boundary and the concept of ``flow" from the graph theory. In particular, we have constructed linearly independent basis vectors of the attractor space using the properties of fundamental cycles of $G_0$. The attractor space can be divided into subspaces $\mathcal{T}$ and $\mathcal{K}$, corresponding to the eigenvalues $\lambda\neq \pm 1$ and $\lambda = 1$, respectively, and an additional subspace which belongs to the eigenvalue $\lambda = -1$. While the basis of $\mathcal{T}$ and $\mathcal{K}$ can be constructed using the same procedure for all finite connected graphs $G_0$, for the last subspace we provided a construction based on case separation, depending on if the graph is bipartite or not and if it involves self-loops.

The use of fundamental cycles have allowed us to considerably expand the results previously found in the literature, which were often limited to planar graphs. The derived construction of the attractor space enables better understanding of the quantum transport models on graphs. In addition, our results have revealed that the attractor space can contain subspaces of eigenvalues different from $\lambda = \pm 1$. In such a case the evolution of the Grover walk with sink will have more complex asymptotic cycle. In fact, the example we have presented in Section~\ref{sec:ex} exhibits an infinite asymptotic cycle, since the phase $\theta$ of the eigenvalues $\lambda_\pm \neq \pm 1$  is not a rational multiple of $\pi$. This feature is missing, e.g., in the Grover walk on dynamically percolated graphs with sinks, where the evolution converges to a steady state.

\section*{Acknowledgement}

ES acknowledges financial supports from the Grant-in-Aid of
Scientific Research (C)  No.~JP19K03616, Japan Society for the Promotion of Science and Research Origin for Dressed Photon.
M\v S is grateful for the financial support from M\v SMT RVO 14000. This publication was funded by the project ``Centre for Advanced Applied Sciences",\\ Registry No. CZ.$02.1.01/0.0/0.0/16\_019/0000778$, supported by the Operational Programme Research, Development and Education, co-financed by the European Structural and Investment Funds and the state budget of the Czech Republic.

\bibliographystyle{jplain}

\begin{thebibliography}{99}


\bibitem{Ambainis2003} 
Ambainis, A.: 
Quantum walks and their algorithmic applications,  
Int. J. Quantum Inf. {\bf 1} (2003) pp.507--518 (2003).

\bibitem{Ambainis}
Ambainis A., Bach E., Nayak A., Vishwanath A., and Watrous J.: One-Dimensional Quantum Walks, Proc. 33rd Annual ACM Symp. on Theory of Computing (2001) pp. 37.

\bibitem{Konno}
Konno N., Namiki T., Soshi T., and Sudbury A.: Absorption problems for quantum walks in one dimension, J. Phys. A: Math. Gen. {\bf 36} (2003) 241.

\bibitem{Bach}
Bach E., Coppersmith S., Goldschen M.P., Joynt R., and Watrous J.: One-dimensional quantum walks with absorbing boundaries, J. Comput. Sys. Sci. {\bf 69} 2004 562.

\bibitem{Yamasaki}
Yamasaki T., Kobayashi H. and Imai H.: Analysis of absorbing times of quantum walks, Phys. Rev. A {\bf 68} (2003) 012302.

\bibitem{IKS}
Inui N., Konno N., and Segawa E.: One-dimensional three-state quantum walk, Phys. Rev. E {\bf 72} (2005) 056112.

\bibitem{SNJ}
Štefaňák M., Novotný J., and Jex I.: Percolation assisted excitation transport in discrete-time quantum walks, New J. Phys. {\bf 18} (2016) 023040.

\bibitem{MNJ}
Mareš J., Novotný J., and Jex I.:
Percolated quantum walks with a general shift operator: From trapping to transport, Phys. Rev. A {\bf 99}, 042129 (2019).

\bibitem{MNSJ}
Mareš J., Novotný J., Štefaňák M., and Jex I.: 
A counterintuitive role of geometry in transport by quantum walks, Phys. Rev. A {\bf 101} (2020) 032113.

\bibitem{MNJ:2020}
Mareš J., Novotný J., and Jex I.: Quantum walk transport on carbon nanotube structures, Phys. Lett. A {\bf 384} (2020) 126302.

\bibitem{DressedPhoton0}
Ohtsu, M., Kobayashi, K., Kawazoe, T.,  Yatsui, T., Naruse, M.: Principles of Nanophotonics (Taylor and Francis, Boca Raton, 2008)

\bibitem{DressedPhoton1}
Nomura, W.,  Yatsui, T., Kawazoe, T.,  Naruse, M and Ohtsu, M.:
Structural dependency of optical excitation transfer via optical near-field interactions between semiconductor quantum dots,
Applied Physics B 100 pp. 181--187
(2010)

\bibitem{Krovi:hypercube}
Krovi H., and Brun T. A.: Hitting time for quantum walks on the hypercube, Phys. Rev. A {\bf 73} (2006) 032341.

\bibitem{Krovi:infhit}
Krovi H., and Brun T. A.: Quantum walks with infinite hitting times, Phys. Rev. A {\bf 74} (2006) 042334.

\bibitem{HKSS}
Higuchi, Yu., Konno, N., Sato, I., Segawa, E., Spectral and asymptotic properties of Grover walks on crystal lattices”, Journal of Functional Analysis 267 (2014) 4197–4235.


\bibitem{FH1}
Feldman, E., Hillery, M.: Quantum walks on graphs and quantum scattering theory, In: Coding Theory and Quantum Computing, edited by D. Evans, J. Holt, C. Jones, K. Klintworth,
B. Parshall, O. Pfister, and H. Ward, Contemp. Math. {\bf 381} (2005) pp.71–96.

\bibitem{FH2}
Feldman, E., Hillery, M.: Modifying quantum walks: a scattering theory approach, J. Phys. A: Math. Theor. {\bf 40} (2007) 11343–11359.

\bibitem{HamSai}
M. Hamano, H. Saigo,
Quantum walk and dressed photon, 
In Proceedings 9th International Conference on Quantum Simulation and Quantum Walks (QSQW 2020), Marseille, France, 20-24/01/2020,
Electronic Proceedings in Theoretical Computer Science 315, pp. 93--99.

\bibitem{HS}
Higuchi, Yu., Segawa, E.: 
Dynamical system induced by quantum walks,
Journal of Physics A: Mathematical and Theoretical {\bf 52} (39) (2019).

\bibitem{HSS}
Higuchi, Yu., Mohamed, S., Segawa, E.:
Electric circuit induced by quantum walk,
Journal of Statistical Physics {\bf 181} pp.603–617 (2020). 

\bibitem{Kato1982}
Kato, T.:
A Short Introduction to Perturbation Theory for Linear Operators,
Springer-Verlag, New York (1982).

\bibitem{Konno2008b} 
Konno, N.: 
Quantum Walks, In: Lecture Notes in Mathematics: {\bf 1954} (2008) pp.309--452, Springer-Verlag, Heidelberg.




\bibitem{NAJ}
Novotný J., Alber G., and Jex I.: Asymptotic evolution of random unitary operator, Cent. Eur. J. Phys. {\bf 8} (2010) pp.1001--1014. 


\bibitem{NYKNO}
Nomura W., Yatsui T., Kawazoe T., Naruse M, and Ohtsu M., 
Appl. Phys. B {\bf 100} (2010) pp.181-187.

\bibitem{Por}
Portugal, R.:
Quantum Walk and Search Algorithms 2nd Ed., 
Springer Nature Switzerland (2018)

\bibitem{R}
Robinson, M.: Dynamical Systems: Stability, Symbolic dynamics, and Chaos, CRC Press
(1995).




\end{thebibliography}

\end{document}